\documentclass{svproc}
\usepackage{url}

\usepackage{dsfont}
\usepackage{amsmath}
\usepackage{color,fancybox,graphicx,url}
\usepackage{float}
\usepackage[caption = false]{subfig}

\newtheorem{prop}{Proposition}
\newtheorem{assumption}{Assumption}

\newcommand{\UCdp}{U_{C,p}}
\newcommand{\UCdc}{U_{C,c}}
\newcommand{\UCda}{U_{0}}

\newcommand{\xs}{{x^*}}

\newcommand{\Bsbl}{\Big [ }
\newcommand{\Bsbr}{\Big ] }

\newcommand{\Bsbrr}{\Big ) }

\newcommand{\cb}[1]{\{#1\}}

\newcommand{\muv}{\mu}
\newcommand{\muc}{\Lambda}

\newcommand{\mucs}{\Lambda^*}

\newcommand{\Bf}[2]{q(\xs(#1)|#1)p(\xs(#1)|#2)}

\newcommand{\threshold}{\mu_0}

\newcommand{\Cc}{\Icen}

\newcommand{\Icen}{I_0}

\newcommand{\possibleStrategy}{\mathcal{S}}

\newcommand{\coreBudget}{M_c}
\newcommand{\leafBudget}{M_p}

\newcommand{\alternative}[1]{r_0 \Bsbl  B_0 e^{- \frac{\alpha}{#1}}-c \Bsbr I(#1)}
\newcommand{\firstOrderAlternative}[1]{\frac{\alpha}{#1^2}e^{ - \frac{\alpha}{#1}}C_0}

\newcommand{\leafAllocation}[1]{\mu_p(#1)}

\newcommand{\coreInteract}[1]{\mu_c(#1)}
\newcommand{\coreInteractOpt}[1]{\mu^*_c(#1)}
\newcommand{\coreInteractVec}{\mu_c}
\newcommand{\coreInteractVecOpt}{\mu^*_c}

\newcommand{\coreUtility}[1]{S(y_c|#1)}
\newcommand{\leafOptUtility}[1]{U_p(\leafRateVec{#1}|\coreRateVec,#1)}

\newcommand{\delayCoreUtility}[2]{r_p \Bsbl p(#1|#2) e^{ - \alpha \left ( \frac{1}{\mu_c(#1)} + \frac{1}{\mu(y_c|#2)} \right)} -  c  \Bsbr I(\mu_c(#1))I(\mu(y_c|#2))}

\newcommand{\delayCoreUtilityR}[2]{r_p \left [ p(z|y) e^{ - \alpha \left ( \frac{1}{#1} + \frac{1}{#2} \right)} -  c  \right ] I(#1)I(#2)}

\newcommand{\delayUtility}[2]{r_p \Bsbl p(#1|#2) e^{ - \frac{\alpha}{\mu(#1|#2)}} -   c\Bsbr I(\mu(#1|#2))}

\newcommand{\optcore}[1]{OPT(\coreRateVec|#1)}
\newcommand{\optperiphery}[2]{OPT(\leafAllocation{#2}|#1,#2)}

\newcommand{\leafSet}{C_p}
\newcommand{\core}{y_c}

\newcommand{\comSet}{C}
\newcommand{\coreRateVec}{\coreInteractVec}

\newcommand{\leafRateVec}[1]{\muv_p(#1)}

\newcommand{\Clw}[1]{\leafSet \backslash\{#1\}}
\newcommand{\Cw}[1]{\comSet \backslash\{#1\}}
\DeclareMathOperator*{\argmax}{argmax}

\newcommand{\strategySpace}{S}

\newcommand{\allstrategySpace}{\possibleStrategy_0}

\begin{document}
\mainmatter              
\title{Structure of Core-Periphery Communities}
\titlerunning{Structure of Core-Periphery Communities}  
%
\author{Junwei Su\inst{1} \and Peter Marbach\inst{1}}
%
%
%
\institute{
Dept. of Computer Science, University of Toronto, Canada}

\maketitle              

\begin{abstract}
It has been experimentally shown that communities in social networks tend to have a core-periphery topology. However, there is still a limited understanding of the precise structure of core-periphery communities in social networks including the connectivity structure and interaction rates between agents. In this paper, we use a game-theoretic approach to derive a more precise characterization of the structure of core-periphery communities.

\end{abstract}

\section{Introduction}\label{sec:intro}

Experimental results have shown that communities in social networks tend to have a core-periphery topology consisting of  two types of agents, core agents and periphery agents, that differ in their objectives for participating in the community~\cite{core_periphery_community_experiment,vlog}.
The objective of periphery agents is to obtain content that is of interest to them. As a result, periphery agents follow other agents in the community to obtain the content that is of most interest to them. 
The objective of the core agents is to attract followers, and attention, from the periphery agents. To achieve their objective, core agents aggregate/collect content from the community and make it available to the periphery agents~\cite{media_aggregate,jordan2012lattice,link_aggregator}.
These two different objectives lead to a community structure where the core agents follow periphery agents in the community in order to collect content,  and the periphery agents connect with the core agents and other periphery agents in order to obtain the content they are interested in~\cite{core_periphery_community_experiment,vlog}.

In this paper, we provide a mathematical model that allows us to derive these structural properties of core-periphery communities in social networks in a formal manner.
The results of our analysis provide a precise characterization of the connectivity structure, and interactions rates, of core-periphery communities.

For our analysis,  we use a game-theoretic framework where we assume that agents in the community make the decision on which other agents to interact with in a manner that maximizes their own objective. The proofs for all the results are provided in the appendix.

\section{Related Work}\label{sec:related_work}



Experimental studies have shown that communities in social networks tend to have a core-periphery topology with two types of agents, core agents and periphery agents~\cite{core_periphery_community_experiment,vlog}, where the core agents collect (aggregate) content from the community, and make it available to the periphery agents.  While experimental studies show that this structure exists, they do not provide a (precise) characterization of the connectivity structure, as well as the interaction rates between the different agents. The goal of this paper is to provide such a characterization. An interesting result from the experimental studies is that core-periphery communities in online networks tend to have a small set of core agents, typically in the order of 1-6 core agents~\cite{core_periphery_community_experiment,vlog}.


Theoretical results on the structural properties of communities were obtained in the context of network formation games\cite{network_formation_survey,free_trade_network,market_share_network,social_econ_game,directed_network_formation,airline_network}.
Jackson and Wolinsky presented one of the first, and most influential, analysis of network formation games~\cite{social_econ_game}. For their analysis,  Jackson and Wolinsky assume that a) agents in the network obtain a benefit from having paths to other agents and b) pay a cost for each direct connection (link) that they have with another agent. The benefit that an agent obtains from another agent is discounted by a factor $\delta^{d}$, where $d$ is the length of the path (distance) between the two agents and $\delta$, $0 < \delta < 1$, is a discount factor.
Assuming bi-directional links, Jackson and Wolinsky show that the star topology is a Nash equilibrium for the game that they consider. The paper by Jackson and Wolinsky makes several important contributions. First, it shows that a game-theoretic model can be used to derive the structural properties of communities. Second, the star-topology of the Nash equilibrium suggests that a  core-periphery topology might indeed naturally emerge as the community topology in social networks. 

Bala and Goyal use in~\cite{directed_network_formation} the model of Jackson and Wolinsky, except that they consider unidirectional links instead of bidirectional links.   For this model, Bala and Goyal show that the star topology again emerges as a Nash equilibrium, and side payments from the periphery agents to the core agent are required for the star topology to emerge as a Nash equilibrium. 

A limitation of the analysis by Jackson and Wolinsky is they assume a homogeneous set of periphery agents. Hegde et al. consider in~\cite{flow_game} a more general model that allows for a heterogeneous population of periphery agents where periphery agents differ in the benefit they obtain from other agents. 
To model the heterogeneous population, Hedge et al. embed agents in a Euclidean space. Agents that are close (in the Euclidean distance) to a given agent provide a higher benefit to the agent compared with agents that are further away.
Assuming that all agents have the same number of connections, Hegde et al. consider the game where agents choose connections to other agents in order to maximize their own benefit. For this model, Hedge et al. show that there exists a Nash equilibrium. However, due to the complexity of the model, Hedge et al. were not able to derive and characterize the structural properties of the Nash equilibrium.

In summary, existing mathematical models are either too simple (as it is the case for~\cite{social_econ_game,directed_network_formation}) and lead to a core-periphery  community structure that does  not accurately reflect the community structures observed in real-life social networks; or they are too complex and can not be used to derive the structural properties of core-periphery communities (as it is  the case for~\cite{flow_game}).
The goal of this paper is to propose a model that is simple enough to characterize the structural properties of a community, and yet is complex enough to lead to results that accurately reflect the community structures that are observed in real-life social networks and can be used to design algorithms for social networks.

\section{Core-Periphery Community}\label{sec:model}
We use the following model for our analysis.

\paragraph{\bf Core-Periphery Community $\comSet$:}
A core-periphery community consists of a set of core agents and periphery agents. To simplify the notation and analysis, we assume that there exists a single core agent $y_c$. This assumption is also motivated by the experimental results which show that core-periphery communities tend to have a small set of core agents, typically in the order of 1-6 core agents~\cite{core_periphery_community_experiment,vlog}. The results that we obtain for a single core agent can be extended to the case of multiple core-agents.
Using this assumption, a core-periphery community $\comSet$ is then given by  a core agent $\core$ and a  set $\leafSet$ of periphery agents, i.e. we have that
$\comSet = \leafSet \cup \{\core \}$.

\paragraph{\bf Periphery Agents $\leafSet$:}
For our analysis, we assume that periphery agents both produce and consume content. In addition, we assume a heterogeneous set of periphery agents, where agents differ in the content (topics) that they are interested in. To model this situation we use a similar approach as in~\cite{flow_game}, and assume a ``topic space'' that specifies how closely two topics are related with each other. The topic space that we consider is given by the interval $I_C = [\Icen - L_C, \Icen + L_C] \subset R$. Each periphery agent is then characterized by its main interest $y \in I_C$, which is the topic that the agent is most interested in. For the content production, we assume that each agent produces content on the topic that is their main interest. For the content assumption, we use the following model. 
The probability that a periphery agent with main interest $y$ is interested in a content produced by an agent with main interest $x$ is given by
	\begin{equation}\label{eq:p(x|y)}
		p(x|y) = f(||x-y||), \qquad x,y \in I_C,
	\end{equation}
where $f:[0,\infty) \mapsto [0,1]$ is a decreasing concave function. 
Note that this definition implies that periphery agents are more interested in content that is produced by agents whose main interest is close to their own main interest.
		
For our analysis we assume that the (main interests of the) periphery agents are "uniformly" distributed in the interval $I_C$, with equal distance $\delta$ between two agents. That is, we assume that the set of periphery agents $C_l$ consists of $K$ agents with main interests $y_k$, $k=1,...,K$, given by
	$$\leafSet = \{y_1,...,y_{K}\} \subset I_C = [\Icen - L_C, \Icen + L_C],$$
	with
	$y_1 = \Icen - L_C$
	and
	$   y_{k+1} = y_k + \delta$, $k=1,...,K-1$,
	where
	$\delta  = \frac{2L_C}{K-1}$.
In the following we identify periphery agents  by their main interest $y$.



\section{Utility of Periphery Agents}\label{sec:utility}
For our analysis, we assume that periphery agents can obtain content from three different sources: a) directly from other periphery agents by following these agents, b) indirectly from the core agent, where the content provided by the core agent is the content that the core agent obtains by following periphery agents in the community, and c) by following content platforms outside the community.~\footnote{For example, users on Twitter will generally also get content from additional content platforms such as other news or other social media sites.}

We use the following notation to characterize the following rates between the agents, and the following rates of periphery agents to content platforms outside the community.

Let $\coreInteract{y}$ be the rate with which core agent $y_c$ follows periphery agent $y \in \leafSet$, and let 
$\coreRateVec = ( \coreInteract{y} )_{y \in \leafSet}$
be following rate vector  of the core agent $y_c$ to all periphery agents $y \in \leafSet$. 

Similarly, let
$\mu(y) = ( \mu(z|y) )_{z \in \comSet \backslash{\{y\}}}$
be the following rate vector  of periphery agent $y \in \leafSet$ to all other agents $z \in \comSet\backslash\{y\}$ in the community.
Furthermore, let  $\lambda(y)$ be the rate with which periphery agent $y$ follows  content platforms outside the community, and let
$\mu_p(y) = (\mu(y),\lambda(y))$
be the overall following rate vector of  periphery agent $y \in \leafSet$.

Finally, let
$ \muc_p  =  ( \leafRateVec{y} )_{y \in \leafSet}$
be the following rate vectors of all periphery agents.

We next define the utilities that periphery agents obtain from following a) other periphery agents directly, b) the core agent, and c) content platforms outside the community. 

\paragraph{\bf Utility from Following a Periphery Agent Directly:}
We first define the utility that a periphery agent $y$ obtains by following another periphery agent $z$ with rate $\mu(z|y)$. 
Suppose that agent $y$ receives a reward of value 1 for each content item that is of interest to agent $y$. Furthermore, suppose  that each content item that agent $y$ receives incurs a processing (reading) cost $c$, $0 < c < 1$. If agent $y$ receives content from agent $z$ with delay $d(z|y)$, then the (expected) utility rate of agent $y$ is given by
\begin{equation}\label{eq:UCd_delay2}
\begin{split}
\UCdp(z|y) = r_p \Bsbl p(z|y) e^{ - \alpha d(z|y)} -   c\Bsbr I(\mu(z|y)),
\end{split}
\end{equation}
where $p(z|y)$ is the probability that a content item of agent $z$ is of interest to agent $y$, $I(\mu(z|y))$ is the indicator function of whether agent $y$ follows agent $z$ and is equal to 1 if $\mu(z|y) > 0$,  $r_p$ is the rate at which $z$ produces content, and $\alpha$ is a given constant that captures how sensitive the content produced by agent $z$ is towards delay. This utility function captures the intuition that the longer the delay $d(z|y)$ is, the lower is the utility of the received content.

For our analysis we define the delay $d(z|y)$ by
\begin{equation*}\label{eq:delay}
    d(z|y) = \frac{1}{\mu(z|y)},
\end{equation*}
where $\mu(z|y)$ is the rate with which agent $y$ follows agent $z$. 
Note that this definition implies that the higher the rate with which  agent $y$ follows agent $z$, the lower the delay $d(z|y)$ will be.

\paragraph{\bf Utility from Following the Core Agent:}
We next define the utility that a periphery agent $y \in \leafSet$ receives from content of periphery agent $z \in \leafSet$, when the content is received through the core agent $y_c$. For this, suppose that the core agent $y_c$ follows periphery agent $z$ with rate $\mu_c(z)$, and periphery agent $y$ follows the core agent $y_c$ with rate $\mu(y_c|y)$. The total delay with which agent $y$ receives content from agent $z$ through $y_c$ is given by
$$d(z|y_c) + d(y_c|y) = \frac{1}{\coreInteract{z}} +  \frac{1}{\mu(y_c|y)}.$$
Using this result, the utility rate of periphery agent $y$ for getting the content of agent $z$ via the core agent $y_c$ is given by
\begin{equation}\label{eq:utility_core}
  \UCdc(z|y) = \delayCoreUtility{z}{y}.
\end{equation}  

\paragraph{\bf Utility from Following other Content Platforms:}
Finally we define the utility that a periphery agent $y \in \leafSet$ obtains by  getting content from other platforms. For this, we assume that  the overall rate (over all content platforms) at which new content items are generated by the other platforms is equal to $r_0 > 0$, and that each content item is of interest to agent $y$ with probability $B_0$. If periphery agent $y$ follows other content platforms with rate $\lambda(y)$, then the corresponding utility rate is given by 
\begin{equation}\label{eq:utility_other}
  \UCda(y) = \alternative{\lambda(y)}.
\end{equation}



\section{Agents' Decisions and  Interactions}\label{sec:problem_formulation}
In this section, we model the interaction among agents in a core-periphery community where we assume that each agent decides on its following rates in order to maximize its own objective function.

\subsection{Core Agent's Decision Problem}\label{ssec:core_agent_model}
Recall from Section~\ref{sec:intro} that the objective of the core agent $\core$ is to attract attention from periphery agents by aggregating/collecting content that is of most interest to the periphery agents~\cite{media_aggregate,jordan2012lattice,link_aggregator}. We formulate the resulting decision problem of the core agent as an optimization problem as follows.

Recall that  $\muc_p$ is the rate allocation vector over the all periphery agents $y \in \leafSet$, and  $\coreRateVec = ( \coreInteract{y} )_{y \in \leafSet}$ is the rate allocation of the core agent $y_c$. Furthermore recall Eq.~\eqref{eq:utility_core} that defines the utility $\UCdc(z|y)$ that periphery agent $y$ obtains from getting content of agent $z$ through the core agent $y_c$. For a given rate allocation  $\muc_p$ of the periphery agents, the decision problem of the core agent $\core$ is given by the following optimization problem  $\optcore{\muc_p}$,

\begin{equation}\label{eq:core_opt}
	\begin{aligned}
		& \underset{\coreRateVec }{\text{maximize}} &&\sum_{y \in \leafSet} \sum_{z \in \Clw{y}} \UCdc(z|y) \\
		& \text{subject to} && \sum_{y \in \leafSet} \coreInteract{y} \leq \coreBudget, \\
		&&& \coreInteract{y} \geq 0, \; y \in C_p,
	\end{aligned}
\end{equation}
where $\coreBudget$ is a constraint on the total rate that the core agent can allocate to follow periphery agents $y \in \leafSet$. This constraint reflects  that the core agent $\core$ has limited resources (time) to follow periphery agents in the community. Note that the optimization problem $\optcore{\muc_p}$ captures the goal of the core agent: the core agent $\core$  wants to use its limited resources to attract attention from the periphery agents by aggregating content that is of most interest to the periphery agents.

\subsection{Periphery Agents' Decision Problem}\label{ssec:periphery_agent_model}
Recall that the objective of a periphery agent $y$ is to obtain as "much content that is of interest as possible". A periphery agent can achieve this goal by following other periphery agents directly, by getting content through the core agent $y_c$, and by getting content from other content platforms. We formulate the resulting decision problem of a periphery agent as follows.

Let $\coreRateVec = ( \coreInteract{y} )_{y \in \leafSet}$ be a given rate allocation of the core agent $y_c$, and let
\begin{equation}\label{eq:periphery_utility}
	\begin{aligned}
		&   \leafOptUtility{y}
		=  
		\sum_{z \in \Clw{y}} \UCdc(z|y) 
		+ \sum_{z \in \Clw{y}} \UCdp(z|y) 
		+ \UCda(y) 
	\end{aligned}
\end{equation}
be the total utility rate that periphery agent $y$ obtains under its rate allocation $\leafRateVec{y}$ and the allocation $\coreRateVec$  of the core agent. For a given rate allocation $\coreRateVec$ of the the core agent,  the decision problem of the periphery agent $y$ is given by the following optimization problem
$\optperiphery{\coreRateVec}{y}$,
\begin{equation}\label{eq:periphery_opt}
	\begin{aligned}
		& \underset{\mu_p(y)}{\text{maximize}}
		& & \leafOptUtility{y}\\
		& \text{subject to} &&\mu(\core|y) + \lambda(y) + \sum_{z \in \Clw{y}} \mu(z|y) \leq \leafBudget, \\
		&&& \mu(z|y), \lambda(y), \mu(\core|y) \geq 0, \; z \in C_p \backslash \cb{y},
	\end{aligned}
\end{equation}
where $\leafBudget>0$ is a constraint on the total rate that periphery  agent $y$ can allocate. To simplify the notation and analysis, we assume that the rate budget $\leafBudget$ is the same for all periphery agents.


\subsection{Nash Equilibrium}\label{ssec:stable}
The optimal solution of the maximization problem  $\optcore{\muc_p}$ of the core agent depends on the given rate allocation  $\muc_p$ of the periphery agents. Similarly, the optimal solution of the maximization problem $\optperiphery{\coreRateVec}{y}$ of periphery agent $y$ depends on the given rate allocation  $\coreInteractVec$ of the core agent. This coupling creates a strategic interaction (game) between the agents in the community. A Nash equilibrium for the resulting game is given as follows.

Let $\muc = (\coreInteractVec,\muc_p)$ be the rate allocation vector that characterizes the rate allocation $\coreInteractVec$ of the core agent, and  the rate allocation vector $\muc_p =  ( \leafRateVec{y} )_{y \in \leafSet}$ over all periphery agents. 
\begin{definition}\label{def:stable_allocation}
	An allocation $\muc^*  = (\coreInteractVecOpt,\muc^*_p)$ is a Nash equilibrium if we have that
	$$\coreInteractVecOpt = \argmax_{\coreRateVec \geq 0}\optcore{\muc^*_p}
	\quad \mbox{ and } \quad
	\mu_p^*(y) = \argmax_{\leafRateVec{y} \geq 0}\optperiphery{\coreInteractVecOpt}{y}.$$
\end{definition}
Definition~\ref{def:stable_allocation} states that under a Nash equilibrium $\muc^* = (\coreInteractVecOpt,\muc^*_p)$ no agent is able to increase the value of their objective function by unilaterally changing their allocation.
In Section~\ref{ssec:nash_equilibrium} we show that there exists a unique Nash equilibrium, and  in Section~\ref{sec:structural_property} we characterize the structural properties of the Nash equilibrium.

\subsection{Existence of Unique Nash Equilibrium}\label{ssec:nash_equilibrium}

For our analysis  we make the following assumptions. 
\begin{assumption}\label{ass:positive_utility}
	For all periphery agents $y \in \leafSet$ we have that
	$$\sum_{z \in \leafSet \backslash \{y\}} \left [ p(z|y)  - c  \right ] > 0
	\qquad \mbox{and} \qquad 
	\sum_{z \in \leafSet \backslash \{y\}} \left [ p(y|z)  - c  \right ] > 0.$$
\end{assumption}

Assumption~\ref{ass:positive_utility} states that if the content of agent $y$ is received by all other agents $z \in C_p\backslash\{y\}$ without delay, then the resulting total utility is positive. 
Similarly, if agent $y$ receives content from all other agents  $z \in C_p\backslash\{y\}$ without delay, then the resulting total utility $y$ is positive.


In addition we make the following assumption for the processing cost $c$. 
\begin{assumption}\label{ass:concave_assumption}
	We have that
	$ c > e^{- 1 }$.
\end{assumption}
Assumption~\ref{ass:concave_assumption} implies that  if agent $y$ follows agent $z$ with rate $\mu(z|y) < \alpha$, then the utility from  content received through agent $z$ will be negative. As a result we have that if agent $y$ follows agent $z$ with a positive rate $\mu(z|y)>0$, then we have that 
$\mu(z|y) > \alpha$.
We then obtain the following results.
\begin{prop}\label{prop:stable_allocation}
	There exists  a unique Nash equilibrium  $\muc^*  = (\coreInteractVecOpt,\muc^*_p)$.
\end{prop}

\section{Structural Properties of Core-Periphery Communities}\label{sec:structural_property}
In this section, we derive the structural properties of a core-periphery community at the Nash equilibrium  $\muc^*  = (\coreInteractVecOpt,\muc^*_p)$.  

\subsection{Condition for Core-Periphery Communities to Emerge}
We first characterize how the rate budget $\coreBudget$ of the core agent, and the rate budgets $\leafBudget$ of the periphery agents, impact the structural properties of the Nash equilibrium. 
We have the following result. 
\begin{prop}\label{prop:sufficient_budget}
	There exists constants $m_c$ and $m_p$  such that if for the rate budget $\coreBudget$ of the core agent and the rate budget $\leafBudget$ of the periphery agents we have that 
	$$\coreBudget > m_c 
	\quad \mbox{ and } \quad 
	\leafBudget > m_p,$$
	then the following is true for the resulting  Nash equilibrium $\muc^* = (\coreInteractVecOpt,\muc^*_p)$.
	For all periphery agents $y \in \leafSet$ we have that
	$$ \coreInteractOpt{y} > 0
	\quad \mbox{ and } \quad
	\mu^*(y_c|y) > 0.$$
\end{prop}  
Proposition~\ref{prop:sufficient_budget} states that if the rate budgets $\coreBudget$ and $\leafBudget$ are high enough then all periphery agents follow the core agent, and the core agent will follow all periphery agents.

Proposition~\ref{prop:sufficient_budget} provides conditions for a core-periphery community to emerge. In a core-periphery community, the core agent collects content from (almost) all periphery agents and makes it available to the periphery agents. In addition, in a core-periphery community (almost) all periphery agents follow the core agent in order to obtain content from the community.  Proposition~\ref{prop:sufficient_budget} states that in order for this structure to emerge, the agents have to be sufficiently interested in getting content and allocated a sufficient amount of time (a sufficiently large rate budget) to sharing online content.

\subsection{Connectivity between Periphery Agents}
We next study the structural properties of how periphery agents follow each other in a core-periphery community. We have the following result.
\begin{prop}\label{prop:periphery_local_interaction} 
	For a  Nash equilibrium $\muc^* = (\coreInteractVecOpt,\muc^*_p)$ as given in Proposition~\ref{prop:sufficient_budget} the following is true. For each periphery agent $y \in C_p$ there exists a threshold $t(y) > 0$ such that
	$I(\mu^*(z|y)) = 1$,
	if, and only if, 
	$p(z|y) > t(y)$.
\end{prop}
Note that the value of $ p(z|y)$ is higher for agents $z$ that are close to agent $y$. As a result, 
Proposition~\ref{prop:periphery_local_interaction} states that each periphery agent $y$ follows other periphery agents $z$ that are  not too far away from $y$. 
Combining this result with Proposition~\ref{prop:sufficient_budget}, Proposition~\ref{prop:periphery_local_interaction} states that core-periphery communities have the structural property that periphery agents follow the core agent, as well as other periphery agents that produce content close to the agents' main interest. This result provides insight into how content is propagated within a core-periphery community. 
In particular, the result implies that content propagates in the following two manners: it spreads (globally) through the core agent within the community, as well as locally through the connection between periphery agents that have similar interests. 

\subsection{Following Rates}
Next, we characterize the following rates between the core agent and periphery agents. We obtain the following result.
\begin{prop}\label{prop:structure}
	For a  Nash equilibrium $\muc^* = (\coreInteractVecOpt,\muc^*_p)$ as given in Proposition~\ref{prop:sufficient_budget} the following is true.
	If for two periphery agents $y,y' \in \leafSet$ we have
	$$||y - \Cc || < ||y' - \Cc ||,$$
	then we have that
	$$\mu^*(y_c|y) > \mu^*(y_c|y')
	\quad \mbox{ and } \quad 
	\mu_c^*(y) > \mu_c^*(y').$$
\end{prop}
Proposition~\ref{prop:structure} states  that periphery agents that are close to the center $\Cc$ of the community have higher interaction rates compared with a periphery agents further away from $\Cc$. More precisely, both the rate $\mu^*(y_c|y)$ with which periphery agent $y$ follows the core agent $y_c$, and the rate $\mu_c^*(y)$ with which the core agent follows periphery agent $y$, is higher for an agent $y$ closer to the center of the community $\Cc$.

Proposition~\ref{prop:structure}  provides  a ``ranking'' or ``ordering'' of periphery agents $y \in \leafSet$ based on how close they are to the center  $\Icen$ of the community. While it is impossible to directly measure how close a periphery agent is with respect to the center of the community, it is possible to measure/estimate the interaction rates of the agent with the core agent. These measurements/estimates can be used in return to infer how close an agent is to the center of the community.


\section{Conclusions}\label{sec:conclusion}
We characterized the structural properties of core-periphery communities using a game-theoretic framework. Assuming that agents allocate a sufficient rate (as given by Proposition~\ref{prop:sufficient_budget}), we obtain the following results:
\begin{enumerate}
	\item[a)] {\bf Connectivity of Core Agents:} Core agents follow  all periphery agents (Proposition~\ref{prop:sufficient_budget}). This confirms the results obtained from experimental studies that core agents serve as a ``hub'' for the community by collecting (aggregating) content and making it available to the other agents in the community~\cite{core_periphery_community_experiment,vlog}.
	\item[b)] {\bf Connectivity of Periphery Agents:} Periphery agents have two types of connections. First, they all follow the core agents (Proposition~\ref{prop:sufficient_budget}). Second, they also follow other periphery agents whose main interest closely matches their interest (Proposition~\ref{prop:periphery_local_interaction}). This result implies that the structure of a core-periphery is not given by a star structure, but has a more complex structure with connections between periphery agents. In addition, this result provides insight into how content propagates within a community (see discussion after Proposition~\ref{prop:periphery_local_interaction}).
	\item[c)] {\bf Interaction Rates:} Periphery agents whose main interest is closer to the center of the community $\Cc$ have higher interaction rates with the core agent compared with agents further away from $\Cc$ (Proposition~\ref{prop:structure}). One possible application of this result is to rank periphery agents with respect to how close their main interest is to the community center $\Cc$ (see discussion after Proposition~\ref{prop:structure}). 
\end{enumerate}

The obtained results provide a mathematical characterization of the structure of core-periphery communities, that can be used to design algorithms.
We are currently using these structural properties to derive community detection algorithms that require only local information, and community-based content recommendation algorithms. The obtained allow us to derive these algorithms in a formal manner, and provide formal performance guarantees.


%

%
%
\bibliographystyle{unsrt}
\bibliography{content/references}

\appendix
\section{PROOF OF PROPOSITION \ref{prop:stable_allocation}}\label{proof:existence}
In this appendix we prove Proposition~\ref{prop:stable_allocation}. To do that, we first establish that the model of Section~\ref{sec:problem_formulation} corresponds to  an exact potential game. Next, we show that the potential function of the game is strictly concave. Finally, we show that the strategy space of each agent is convex and and compact. Proposition~\ref{prop:stable_allocation} then follows from Theorem 2 in~\cite{unique_concave_potential} which states that a potential game with strictly concave potential function, and convex and compact strategy space, has a unique Nash equilibrium.

Recall that $\muc = (\coreRateVec,\muc_p)$ is the set of allocation vectors of all agents in the community, where $\coreRateVec$ is the allocation of core agent $y_c$, and $\muc_p = \{\mu_p(y)\}_{y \in C_p}$ is the set of allocation of periphery agents. To simplify the notation, we use  $\muc_{-z} = \{\mu_p(y)\}_{y \in C_p \backslash \{z\}}$ to denote the set of allocation of periphery agents except for periphery agent $z$. Furthermore, recall that the utility function of core agent is given by

\begin{equation*}
\begin{split}
U(y_c|\muc_p,\coreRateVec) = 
 \sum_{y \in \leafSet } \sum_{z \in \Clw{y}}   \delayCoreUtility{z}{y},
\end{split}
\end{equation*}
and the utility function of periphery agents is defined as
\begin{equation*}
\begin{split}
U(y|\coreRateVec,\leafRateVec{y}) = & \sum_{z \in \Clw{y}} \delayCoreUtility{z}{y}  \\
&+ \delayUtility{z}{y} + \alternative{\lambda(y)}.
\end{split}
\end{equation*}

The following lemma states that the model of Section~\ref{sec:problem_formulation} corresponds to  an exact potential game.

\begin{lemma}\label{lemma:potential_game}
	The interaction among agents form a potential game and the potential function is given by 
	\begin{equation*}
	\begin{aligned}
	G(\muc) = 
	&\sum_{y \in C_p} \sum_{z \in C_p \backslash\{y\}}  \delayCoreUtility{z}{y}\\
	&+ \sum_{y \in C_p} \delayUtility{z}{y}  + \alternative{\lambda(y)}.
	\end{aligned}
	\end{equation*}
\end{lemma}

\begin{proof}
	To show the result of this lemma, we need to show that if any agent (core or periphery agent) changes its strategy/allocation, then the difference in the potential function $G(\muc)$ is equal to the difference in the utility that the agent obtains. In other word, we need to establish the following two properties.
	
First, we have show that if core agent $y_c$ changes its allocation from $\mu_c'$ to $\mu_c''$, then we have that
	\begin{equation}\label{eq:core_change}
		\begin{aligned}
			&G((\mu_c',\muc_p)) - G((\mu_c'',\muc_p)) 
		   = U(y_c|\muc_p,\mu'_c) - U(y_c|\muc_p,\mu''_c).
		\end{aligned}
	\end{equation}
        
Second, we have to show that if periphery agent $y$ changes its allocation from $\mu'_p(y)$ to $\mu''_p(y)$, then we have that
	\begin{equation}\label{eq:leaf_change}
	\begin{aligned}
		&G((\mu_c,\muc_{-y}\cup\{\mu'_p(y)\})) - G((\mu_c,\muc_{-y}\cup\{\mu''_p(y)\})) 
	    =U(y|\mu'(y),\mu_c) - U(y|\mu''(y),\mu_c).
	\end{aligned}
	\end{equation}

We first establish Eq.~\eqref{eq:core_change}. Using the definition of $G(\muc)$, we rewrite $G((\mu_c',\muc_p)) - G((\mu_c'',\muc_p))$ as 
	
	\begin{equation*}
		\begin{aligned}
		&G((\mu_c',\muc_p)) - G((\mu_c'',\muc_p)) = \\
		& \sum_{y \in C_p} \sum_{z \in \Clw{y}} \delayCoreUtilityR{\mu_c'(z)}{\mu(y_c|z)}\\
		& + \sum_{y \in C_p} \delayUtility{z}{y} + \sum_{y \in C_p} \alternative{\lambda(y)}  \\			
		& - \sum_{y \in C_p} \sum_{z \in \Clw{y}} \delayCoreUtilityR{\mu_c''(z)}{\mu(y_c|z)}\\
		& - \sum_{y \in C_p} \delayUtility{z}{y} + \sum_{y \in C_p} \alternative{\lambda(y)}.  \\		
				\end{aligned}.
		\end{equation*}
Using this expression, it follows that
	\begin{equation*}
	  \begin{aligned}
               &G((\mu_c',\muc_p)) - G((\mu_c'',\muc_p)) = \\
		& \sum_{y \in C_p} \sum_{z \in \Clw{y}} \delayCoreUtilityR{\mu_c'(z)}{\mu(y_c|z)} \\
	    &-\sum_{y \in C_p} \sum_{z \in \Clw{y}}\delayCoreUtilityR{\mu_c''(z)}{\mu(y_c|z)} \\
          	=& U(y_c|\muc_p,\mu_c') - U(y_c|\muc_p,\mu_c'').  
		\end{aligned}
	\end{equation*}

	Next, we establish Eq.~\eqref{eq:core_change}. We obtain that
	\begin{equation*}
	\begin{aligned}
		&G((\mu_c,\muc_{-y}\cup\{\mu'_p(y)\})) - G((\mu_c,\muc_{-y}\cup\{\mu''_p(y)\})) = \\
		&  \sum_{y \in \Clw{y_k}} \sum_{z \in \Clw{y_k,y}} \delayCoreUtility{z}{y}\\
		& + \delayUtility{z}{y} + \alternative{\lambda(y)}  \\
		& + \sum_{z \in \Clw{y_k}}  r_p[p(z|y_k) e^{ - \frac{\alpha}{\mu_c(z)} - \frac{\alpha}{\mu'(y_c|y_k)} }  -  c ] I(\mu'(y_c|y_k))I(\mu_c(z))\\
		&+  r_p [ p(z|y_k) e^{ - \frac{\alpha}{\mu'(z|y_k)}} -   c ]I(\mu'(z|y_k)) 
		+ \alternative{\lambda'(y_k)} \Bsbrr\\
		& -  \sum_{y \in \Clw{y_k}} \sum_{z \in \Clw{y_k,y}} \delayCoreUtility{z}{y} \\
		& - \delayUtility{z}{y} + \alternative{\lambda(y)} \\
		& - \sum_{z \in \Clw{y_k}}  r_p [p(z|y_k) e^{ - \frac{\alpha}{\mu_c(z)} - \frac{\alpha}{\mu''(y_c|y_k)} }  -  c ] I(\mu''(y_c|y_k))I(\mu_c(z))  \\
		&- r_p [ p(z|y_k) e^{ - \frac{\alpha}{\mu''(z|y_k)}} -   c ]I(\mu''(z|y_k))
		 - \alternative{\lambda''(y_k)}.\\
	\end{aligned}
	\end{equation*}

It then follows that 
	
	\begin{equation*}
	  \begin{aligned}
            	&G((\mu_c,\muc_{-y}\cup\{\mu'_p(y)\})) - G((\mu_c,\muc_{-y}\cup\{\mu''_p(y)\})) = \\
		  & \sum_{z \in \Clw{y_k}}  r_p [p(z|y_k) e^{ - \frac{\alpha}{\mu_c(z)} - \frac{\alpha}{\mu'(y_c|y_k)} }  -  c ] I(\mu'(y_c|y_k))I(\mu_c(z))  \\
		&+ r_p [ p(z|y_k) e^{ - \frac{\alpha}{\mu'(z|y_k)}} -   c ]I(\mu'(z|y_k))\\ 
		&+ \alternative{\lambda'(y_k)}\\
		&-  \sum_{z \in \comSet \backslash\{y_k\}}  r_p  [p(z|y_k) e^{ - \frac{\alpha}{\mu_c(z)} - \frac{\alpha}{\mu''(y_c|y_k)} }  -  c ] I(\mu''(y_c|y_k))I(\mu_c(z))\\
		&- r_p [ p(z|y_k) e^{ - \frac{\alpha}{\mu''(z|y_k)}} -   c ]I(\mu''(z|y_k))\\ 
	    &- \alternative{\lambda''(y_k)}\\
            &= U(y|\mu'(y),\mu_c) - U(y|\mu''(y),\mu_c).
		\end{aligned}
	\end{equation*}
	
	This completes the proof of the lemma.
\end{proof}

We next show  that the potential function $G(\muc)$ is strictly concave. To do that, we define more precisely the strategy space of each agent. To do that, let rate threshold $\threshold$ be such that
$$e^{- \frac{\alpha}{\threshold}} = c.$$
Note that for the rate threshold  $\threshold$, we have that if 
$$ \mu(z|y) > 0,$$
and agent $y$ follows agent $z$ with a positive rate, 
        then we have that
        $$ \mu(z|y) \geq \threshold.$$
        To see this, note that by the definition of the utility of agent $z$ we have for $0 < \mu(z|y) < \threshold$ that the utility that agent $y$ obtains from following agent $z$ is negative. As a result, in this case agent $y$ is better off not to follow agent $z$ at all, and set $\mu(z|y) = 0$, and receive a utility equal to 0 for the content of agent $z$. By Assumption~\ref{ass:concave_assumption}, we have that
$$ \threshold > \alpha.$$

Without loss of generality, we can then consider rate vectors $\leafAllocation{y}= ( \mu(z|y) )_{z \in \comSet \backslash{\{y\}}}$ in the set (strategy space) $\strategySpace_p$ given by
$$\strategySpace_p = A^{K+1},$$
where $K$ is the number of periphery agents, and the  set $A$ is given by
$$ A = \{0\} \cup [\threshold,\leafBudget].$$
Similarly, the strategySpace of the core agent is given by
$$\strategySpace_c =  B^K,$$
where the set $B$ is given by
$$B = \{0\} \cup [\threshold,\coreBudget].$$
The strategy space over all agents (core and periphery agents) is then given by
$$\allstrategySpace = \strategySpace_c \times \strategySpace_p^K,$$
and to analyze the Nash equilibrium it suffices to consider rate allocations
$$\muc \in \allstrategySpace.$$
Note that the strategy space $\strategySpace$ is convex and compact.

\begin{lemma}\label{lemma:potential_game_convexity}
	The potential function $G(\muc)$ of Lemma \ref{lemma:potential_game} is a strictly concave on $\allstrategySpace$.
\end{lemma}

\begin{proof}
	We rewrite the potential function $G(\muc)$ as 
	\begin{equation*}
	\begin{aligned}
	&G(\muc) = 
	 \sum_{y \in C_l} \sum_{z \in \Clw{y}}  \Bsbl f_1(\mu_c(z), \mu(y_c|y)) + f_2(\mu(z|y))   + f_3(\lambda(y)) \Bsbr,
	\end{aligned}
	\end{equation*}
	where 

	\begin{equation}\label{eq:core_term}
	\begin{aligned}
		&f_1(\mu_c(z), \mu(y_c|y)) = 
	 &\delayCoreUtility{z}{y},
	\end{aligned}
	\end{equation}
	and
	\begin{equation}\label{eq:leaf_term}
		f_2(\mu(z|y)) 		= \delayUtility{z}{y},
	\end{equation}
	as well as
	\begin{equation}\label{eq:alternative_term}
		f_3(\lambda(y)) = \alternative{\lambda(y)}.
	\end{equation}
	
	To show $G(\muc)$ is strictly concave, it suffices to show that the functions  $f_1(\mu_c(z), \mu(y_c|y))$, $f_2(\mu(z|y)), f_3(\lambda(y))$  are strictly concave under Assumption~\ref{ass:concave_assumption}. To do this, we first show that the Hessian of the  function $f_1(\mu_c(z), \mu(y_c|y))$ is negative definite under Assumption~\ref{ass:concave_assumption}. Next, we show that the second derivatives of the function $f_2(\mu(z|y))$ and $f_3(\lambda(y))$ are negative under Assumption~\ref{ass:concave_assumption}.

        We will with showing the Hessian of the function $f_1(\mu_c(z), \mu(y_c|y))$ is negative definite under Assumption~\ref{ass:concave_assumption}.  As $p(z|y)$, $r_p$, and $c$, are positive constants, in order to show this it suffices to establish that the Hessian of the function $ e^{ - \frac{\alpha}{\mu_c(z)} - \frac{\alpha}{\mu(y_c|z)} }$ is negative definite.

    The second derivatives of the function $e^{ - \frac{\alpha}{\mu_c(z)} - \frac{\alpha}{\mu(y_c|z)} }$  are given by	
$$\frac{d^2}{d\mu_c(z)d\mu(y_c|z)}e^{ - \frac{\alpha}{\mu_c(z)} - \frac{\alpha}{\mu(y_c|z)} }= \frac{\alpha^2}{\mu_c^2(z) \mu^2(y_c|z)}e^{ - \frac{\alpha}{\mu_c(z)} - \frac{\alpha}{\mu(y_c|z)} },$$
    and
    $$\frac{d^2}{d^2\mu_c(z)}e^{ - \frac{\alpha}{\mu_c(z)} - \frac{\alpha}{\mu(y_c|z)} }
    = \frac{\alpha}{\mu_c(z)^3} \left [ \frac{\alpha}{\mu_c(z)} - 2 \right ]  e^{ - \frac{\alpha}{\mu_c(z)} - \frac{\alpha}{\mu(y_c|z)} },$$
    as well as
	$$\frac{d^2}{d^2\mu(y_c|z)}e^{ - \frac{\alpha}{\mu_c(z)} - \frac{\alpha}{\mu(y_c|z)} } = \frac{\alpha}{\mu(y_c|z)^3} \left [ \frac{\alpha}{\mu(y_c|z)} - 2 \right ] e^{ - \frac{\alpha}{\mu_c(z)} - \frac{\alpha}{\mu(y_c|z)} }.$$

        It then follows that the Hessian $H$ of the function $e^{ - \frac{\alpha}{\mu_c(z)} - \frac{\alpha}{\mu(y_c|z)} }$ is given by 
        \begin{equation*}
	\begin{aligned}
	 &H = \begin{pmatrix}
	h11 & h12 \\
	h21 & h22 
	\end{pmatrix},
	\end{aligned}
	\end{equation*}
where
	\begin{equation*}
		\begin{split}
		 h11 &= 	\frac{\alpha}{\mu_c(z)^3} \left [ \frac{\alpha}{\mu_c(z)} - 2 \right ]  e^{ - \frac{\alpha}{\mu_c(z)} - \frac{\alpha}{\mu(y_c|z)} }\\
		 h12 &= \frac{\alpha^2}{\mu_c^2(z) \mu^2(y_c|z)}e^{ - \frac{\alpha}{\mu_c(z)} - \frac{\alpha}{\mu(y_c|z)} }\\
		 h21 &= \frac{\alpha^2}{\mu_c^2(z) \mu^2(y_c|z)}e^{ - \frac{\alpha}{\mu_c(z)} - \frac{\alpha}{\mu(y_c|z)} }\\
		 h22 &=  \frac{\alpha}{\mu(y_c|z)^3} \left [ \frac{\alpha}{\mu(y_c|z)} - 2 \right ]  e^{ - \frac{\alpha}{\mu_c(z)} - \frac{\alpha}{\mu(y_c|z)} }. 
		\end{split}
	\end{equation*}
        
Using the definition of the strategy space $\strategySpace$, in order to show that the Hessian $H$ is negative definite we can consider the case where either we have that
        $$\mu(z|y) = 0,$$
        or
        $$ \mu(z|y) \geq \threshold.$$
        Similarly, it suffices to consider the case where 
        $$\mu_c(z) = 0,$$
        or
        $$ \mu_c(z) \geq \threshold.$$
        
	We first consider the case where
        $\mu_c(z) \geq \threshold$ and $\mu(y_c|z) \geq \threshold$.
	In this case, we have that
	\begin{equation*}
	\begin{aligned}
	& \begin{pmatrix} \mu_c(z), \mu(y_c|z) \end{pmatrix}H \begin{pmatrix}	\mu_c(z)\\\mu(y_c|z)
	\end{pmatrix} \\
	  & = e^{ - \frac{\alpha}{\mu_c(z)} - \frac{\alpha}{\mu(y_c|z)} }
          \alpha \left [ \frac{1}{\mu_c(z)}\left ( \frac{\alpha}{\mu_c(z)}
            - 2 \right ) 
	  +  \frac{1}{\mu(y_c|z)} \left ( \frac{\alpha}{\mu(y_c|z)} - 2 \right ) + \frac{2\alpha^2}{\mu_c(z) \mu(y_c|z)} \right ] \\
        & <  e^{ - \frac{\alpha}{\mu_c(z)} - \frac{\alpha}{\mu(y_c|z)} }\alpha \left [ \frac{2\alpha^2}{\mu_c(z) \mu(y_c|z)} 
	 -  \frac{1}{\mu(y_c|z)} - \frac{1}{\mu_c(z)} \right ] \\
	& = e^{ - \frac{\alpha}{\mu_c(z)} - \frac{\alpha}{\mu(y_c|z)} }\alpha \left [ \frac{2\alpha-\mu_c(z)-\mu(y_c|z)}{\mu_c(z) \mu(y_c|z)} \right ].\\
	\end{aligned}
	\end{equation*}
        By Assumption~\ref{ass:concave_assumption} we have that
        $$\threshold > \alpha,$$
        and it follows that for $\mu_c(z) \geq \threshold$ and $\mu(y_c|z) \geq \threshold$ we have that
$$  \begin{pmatrix} \mu_c(z), \mu(y_c|z) \end{pmatrix}H \begin{pmatrix}	\mu_c(z)\\\mu(y_c|z)
	\end{pmatrix} < 0.$$      
This implies that the function $f_1(\mu_c(z), \mu(y_c|y))$ is strictly concave for the case where  $\mu_c(z) \geq \threshold$ and $\mu(y_c|z) \geq \threshold$.

Using the same argument, we can show that $f_1(\mu_c(z), \mu(y_c|y))$ is strictly concave for the case where  $\mu_c(z)\geq \threshold$ and  $\mu(y_c|y)=0$, and the case  where  $\mu_c(z)=0$ and $\mu(y_c|y)\geq \threshold$.
It follow that the  function $f_1(\mu_c(z), \mu(y_c|y))$ is strictly concave on $\strategySpace$. 

In addition, using the same argument, we can show that the functions $f_2(\mu(z|y))$ and  $f_3(\lambda(y))$, are strictly concave on $\strategySpace$. This completes the proof of the lemma.
	 
\end{proof}

Using  Lemma~\ref{lemma:potential_game} and Lemma~\ref{lemma:potential_game_convexity}, the result of Proposition~\ref{prop:stable_allocation} follows directly from Theorem 2 in~\cite{unique_concave_potential} which states that a potential game with strictly concave potential function, and convex and compact strategy space, has a unique Nash equilibrium.

\section{PROOF OF PROPOSITION \ref{prop:sufficient_budget}}\label{proof:sufficient_budget}

To prove Proposition ~\ref{prop:sufficient_budget}, we first show that core agent allocates positive rate to all periphery agents once its rate budget $\coreBudget$ is large enough. 
	
\begin{lemma}\label{lemma:core_positive}
  There exists a finite $b_1>0$ such that for  $\coreBudget > b_1$, we have that
  $$\coreInteractOpt{y} > 0, \qquad y \in \leafSet,$$
  where $\coreInteractOpt{y}$, $y \in \leafSet$, is the allocation of the core agent at the Nash equilibrium $\mucs$.
\end{lemma}

\begin{proof}
	We will prove the lemma by contradiction. 
	Suppose that the core agent $y_c$ does not follow an agent $y$, no matter how large the rate budet $M_c$ is. There are two possible reasons for this: Case 1) the core agent $y_c$ does not interact with any periphery agent and we have that
        $$\sum_{y \in \leafSet} \coreInteractOpt{y} =0,$$
        and Case 2) the core agent $z$ follows at least one periphery agent $y'$ with a positive rate and we have that
$$\lim_{\coreBudget \to \infty} \coreInteractOpt{y'} = \infty.$$        

We first consider the case where $\sum_{y \in \leafSet} \coreInteractOpt{y} =0$.  
By Assumption~\ref{ass:positive_utility}, we have that
	\begin{equation}\label{eq:production_utility}
	\begin{aligned}
	\sum_{z \in \Clw{y}} p(y|z) - c > 0.
	\end{aligned}
	\end{equation}	
	To simplify the notation, let
        $$R(y) = \sum_{z \in \Clw{y}} p(y|z).$$
        With this, we can rewrite Eq~\eqref{eq:production_utility} 
	$$R(y) - (K-1)c,$$
	where $K$ is number of periphery agents in $\leafSet$. By Assumption~\ref{ass:positive_utility}, there exists a finite positive constant $M$ such that,
	$$R(y)e^{-\frac{\alpha}{M}} - (K-1)c > 0$$
It the follows that if $\coreBudget > M$, then the core $y_c$ gets a  positive utility by following agent $y$ and the rate allocation such that
        $$\sum_{y \in \leafSet} \coreInteractOpt{y} = 0$$
        is not optimal. This leads to a contraction to the assumption that the allocation  $\coreInteractVecOpt$ is optimal. 
        
        We next consider the case where there exists a agent $y'$ such that
 $$\lim_{\coreBudget \to \infty} \coreInteractOpt{y'} = \infty.$$       
By Assumption~\ref{ass:positive_utility}, there exists a finite positive constant $M$ such that 
	$$R(y)e^{-\frac{\alpha}{M}} - (K-1)c = \Delta > 0.$$	
Without loss of generality, we assume that the $\coreBudget$ is large enough such that $ \coreInteractOpt{y'} > M$. Suppose that the core agent reduces the rate to agent $y'$ by $M$, and consider the resulting difference in the utility given by
	\begin{equation*}
		R(y')e^{-\frac{\alpha}{ \coreInteractOpt{y'}}} - R(y')e^{-\frac{\alpha}{ \coreInteractOpt{y'}-M}}.
	\end{equation*} 
	By Lemma~\ref{lemma:potential_game_convexity}, the function $e^{-\frac{\alpha}{x}}$ is strictly concave for $ x \in \strategySpace_c$, and we have that 
	\begin{equation*}
	R(y')e^{-\frac{\alpha}{ \coreInteractOpt{y'}}} - R(y')e^{-\frac{\alpha}{ \coreInteractOpt{y'}-M}} < M \frac{\alpha}{ \coreInteractOpt{y'}^2}R(y') e^{-\frac{\alpha}{ \coreInteractOpt{y'}}} 
	\end{equation*}
	As we have that 
	$$\lim_{ \coreInteractOpt{y'} \to \infty} \frac{\alpha}{ \coreInteractOpt{y'}^2}R(y') e^{-\frac{\alpha}{ \coreInteractOpt{y'}}} = 0,$$
it follows that there exists a finite $\mu^*$ such that if $\coreInteractOpt{y'} > \mu^*$, then we have that
	\begin{equation*}
	\begin{aligned}
	&	R(y')e^{-\frac{\alpha}{ \coreInteractOpt{y'}}} - R(y')e^{-\frac{\alpha}{ \coreInteractOpt{y'}}}
	< \Delta.
	\end{aligned}
	\end{equation*}
        As by assumption we have that
        $$\lim_{\coreBudget \to \infty} \coreInteractOpt{y'} = \infty,$$ 
that there exists a finite constant $M_c^*$, such that for  $\coreBudget > M_c^*$, then the core agent $y$ can increase its  utility by setting $\mu_c(y) = M$ and $\mu_c(y') = \coreInteractOpt{y'}  - M$.
This leads to a contraction to the assumption that the allocation  $\coreInteractVecOpt$ is optimal. 

This completes the proof of the lemma. 
 	
\end{proof}
Using the same argument that we used to prove Lemma~\ref{lemma:core_positive}, we have that		
\begin{lemma}\label{lemma:core_infinite}
We have that 
  $$\lim_{\coreBudget \to \infty} \coreInteractOpt{y} = \infty, \qquad y \in \leafSet,$$
  where $\coreInteractOpt{y}$, $y \in \leafSet$, is the allocation of the core agent at the Nash equilibrium $\mucs$.
\end{lemma}

Let $S(y|\mu_c)$ be given by
	\begin{equation}\label{eq:core_utility_simp}
	\begin{aligned}
		&S(y|\mu_c) = \sum_{z \in \Clw{y}} p(z|y) e^{ - \frac{\alpha}{\mu_c(z)} }- c.
	\end{aligned}
	\end{equation}
We then have the following result.
	
\begin{lemma}\label{lemma:bmin2}
		There exists a finite $b_2>0$ such that for  $\coreBudget > b_2$ we have that
		$$S(y|\coreInteractVecOpt)) > 0, \qquad y \in \leafSet,$$
where $\coreInteractVecOpt$ is the allocation of the core agent at the Nash equilibrium $\mucs$.
	\end{lemma}
	
	\begin{proof}
	 By Assumption~\ref{ass:positive_utility}, we have that 
	 $$K(y) = \sum_{z \in \Clw{y}} p(z|y) - c> 0, \qquad y \in \leafSet.$$
         Let $K$ be given by
         $$ K = \min_{y \in \leafSet} K(y).$$
         Note that we have that $K>0$. 
		
	Furthermore, by Lemma~\ref{lemma:core_infinite} we have that
		\begin{equation*}
		\lim_{\coreBudget \xrightarrow{} \infty}  S(y|\coreInteractVecOpt) = \sum_{z \in C_l \backslash \cb{y}} p(z|y) -c.
		\end{equation*}
		
		Therefore, for very $\epsilon > 0$,  there exists a  $b_2$ such that for  $\coreBudget > b_2$, we have
		
		\begin{equation*}
		\begin{split}
		\sum_{z \in C_l \backslash \cb{y}} (p(z|y)-c) - S(y|\coreInteractVecOpt) &< \epsilon, \qquad y \in \leafSet.\\
		\end{split}
		\end{equation*}
                In particular, for $\epsilon = K/2$, there exists a $b_2$ such that for  $\coreBudget > b_2$, we have
    \begin{equation*}
		\begin{split}
		\sum_{z \in C_l \backslash \cb{y}} (p(z|y)-c) - S(y|\coreInteractVecOpt) &< \epsilon = \frac{K}{2}, \qquad y \in \leafSet,\\
		\end{split}
		\end{equation*}            
and we obtain that                
		\begin{equation*}
		\begin{split}
		 S(y_c|\coreInteractVecOpt) > \sum_{z \in C_l \backslash \cb{y}} p(z|y) - c - \frac{K}{2} = K(y) - \frac{K}{2} \geq \frac{K}{2} > 0, \qquad y \in \leafSet.\\
		\end{split}
		\end{equation*}
	\end{proof}

        Using  Lemma ~\ref{lemma:bmin2}, we obtain the following result.

        \begin{lemma}\label{lemma:periphery_sufficient}
There exists constants $m_c$ and $m_p$  such that if
$$\coreBudget > m_c 
\quad \mbox{ and } \quad 
\leafBudget > m_p,$$
then under the Nash equilibrium $\muc^* = (\coreInteractVecOpt,\muc^*_p)$ we have for all periphery agents $y \in \leafSet$ that
$$  \mu^*_p(y_c|y) > 0.$$
        \end{lemma}  

\begin{proof}        

	Let $b_2$ be given as in Lemma ~\ref{lemma:bmin2}, and let $\coreBudget > b_2$. Using Lemma~\ref{lemma:bmin2}, there exists a finite positive $M$ such that
	
	$$S(y|\coreInteractVecOpt)e^{-\frac{\alpha}{M}}-c = \Delta >  0.$$

We can then prove the lemma by contradiction, using the same argument as given in the proof for Lemma~\ref{lemma:core_positive}.     
   
\end{proof}	
Proposition~\ref{prop:sufficient_budget} then follows directly from Lemma~\ref{lemma:core_positive} and Lemma~\ref{lemma:periphery_sufficient}.

\section{PROOF OF PROPOSITION \ref{prop:periphery_local_interaction}}\label{proof:periphery_local_interaction}

	We established in Appendix ~\ref{proof:existence} that the interactions between core and periphery agents is characterized by an exact potential function which is concave. This means that the optimal rate allocations of each agent has to satisfy the first order conditions.

	Let $J(\mu(z|y))$ denote the partial derivative between two agents $z$ and $y$ and it is given by,
	\begin{equation*}
	J(\mu(z|y)) = \frac{\alpha r_p }{\mu^2(z|y)}e^{ - \frac{\alpha}{\mu(z|y)}}   p(z|y)
	\end{equation*}
	The derivative of $J(\mu(z|y))$ is given as follow,
	\begin{equation*}
	\frac{dJ(\mu(z|y))}{d\mu(z|y)} = \frac{\alpha}{\mu^3(z|y)}(\frac{\alpha}{\mu(z|y)} - 2) e^{ - \frac{\alpha}{\mu(z|y)}} p(z|y)
	\end{equation*}
	
	As we can see from above that, the maximum of $J(\mu(z|y))$ is obtained at $\mu(z|y) = \frac{\alpha}{2}$ and $J(\mu(z|y))$ would monotonic decrease afterwards.
	
	Next, recall the utility of core agent $S(y|y_c)$ which is given by Eq~\eqref{eq:core_utility_simp}. let $\mu(y_c|y) = \mu_y^*$ denote the optimal allocation of agent $y$ to the core under the Nash equilibrium given by Proposition~\ref{prop:sufficient_budget}. 
	
	By first order condition, we know that if $y$ interact with agent $z$, then the partial derivative converge to the same constant as the term of core agent. Since we know that the partial derivative of each term is monotonic decreasing and achieve maximum at $\frac{\alpha}{2}$. Therefore, agent $y$ interacts with agent $z$ if and only if the maximum partial derivative exceed the partial derivative with respect to the core agent. In other word, we need the following condition,
	\begin{equation*}
	\frac{\alpha}{(\frac{\alpha}{2})^2}e^{-2} p(z|y) > S(y|y_c)\frac{\alpha}{(\mu_y*)^2}e^{-\frac{\alpha}{\mu_y^*}},
	\end{equation*}
 Therefore, for agent $y$ to follow agent $z$, we need 
	\begin{equation*}
	p(z|y) > S(y_c|y)\frac{\alpha}{(\mu_y*)^2}e^{-\frac{\alpha}{\mu_y^*}}e^2\frac{4}{\alpha},
	\end{equation*}
	which is equivalent to 
	\begin{equation*}
	\Bf{z}{y} > S(y_c|y)\frac{4}{(\mu_y*)^2}e^{-\frac{\alpha}{\mu_y^*}}e^2
	\end{equation*}
	Therefore, $t(y) = S(y_c|y)\frac{4}{(\mu_y*)^2}e^{-\frac{-\alpha}{\mu_y^*}}e^2$ is the threshold for the periphery-periphery interaction to happen. 

\section{PROOF OF PROPOSITION \ref{prop:structure}}\label{proof:structure}
In this appendix, we will prove Proposition \ref{prop:structure}. To show the result of this proposition, we first show that the utility function of core agent is concave, and if the periphery agent allocate the way given in the proposition, it is unique and optimal for the core agent to allocate the way given in the proposition. Then, we show similar result for the periphery agents. Using these results and Proposition~\ref{prop:stable_allocation}, we construct an iterative update process that converge to the unique stable allocation, and show that the structural properties given in the proposition would preserve in each iteration. Therefore, the stable allocation in the end has the same structural property. Next, we first establish the result for the core agent.

%
%
%

\subsection{Core Agent Allocation}
Here, we will established that the core agent optimization problem is concave under Assumption \ref{ass:concave_assumption}, and structural property of core agent's optimal respond when periphery agents' allocation has property described in the proposition. We will use these result to show that if the allocation of periphery agents satisfy property in the proposition, then it is uniquely optimal for the core agent to respond in a way with structure described in the Proposition \ref{prop:stable_allocation}.

Recall the core agent allocation problem 
\begin{equation*}
\begin{split}
&\max_{\mu(y_c)}\sum_{y \in C} \sum_{z \in \Clw{y}}  \\
& \hspace{0.5in} \delayCoreUtility{z}{y}
\end{split}
\end{equation*}
subject to
$$ \sum_{y \in C} \mu_c(y) \leq \coreBudget,$$
$$\mu_c(y) \geq 0.$$

\begin{lemma}\label{lemma:core_concavity}
    Under Assumption \ref{ass:concave_assumption}, the optimization problems that characterize core agent's allocation is concave.
\end{lemma}
\begin{proof}

Recall that $U(y_c|\mu(y_c),\muc_p)$ is given as follows
\begin{equation*}
\begin{split}
U(y_c|\muc_p,\coreRateVec) = 
 \sum_{y \in C_p } \sum_{z \in \Cw{y}}  
 \delayCoreUtility{z}{y}.
\end{split}
\end{equation*}
We can rewrite $U(y_c|\mu(y_c),\muc_p)$ as
\begin{equation*}
\begin{split}
& U(y_c|\muc_p,\coreRateVec) = \sum_{y \in C_p } \sum_{z \in \Cw{y}}  f_1(\mu_c(z), \mu(y_c|y)),
\end{split}
\end{equation*}
where $f_1(\mu_c(z), \mu(y_c|y))$ is given in Eq.~\eqref{eq:core_term}. And we proved in Proposition~\ref{prop:stable_allocation} that $f_1(\mu_c(z), \mu(y_c|y))$ is indeed concave under Assumption~\ref{ass:concave_assumption}. Therefore $U(y_c|\muc_p,\coreRateVec)$ is sum of concave function, and therefore concave. This completes the proof of this lemma.
\end{proof}

Now that we have established the objective function of the core agent is concave. The optimal allocation of core agent has to satisfy the  optimal condition, and the partial derivative of positive allocation is computed as following

\begin{equation*}
\begin{split}
\frac{dU(y_c|\mu(y_c),\muc_p)}{d\mu_c(y)}&=\frac{\alpha }{\mu_c^2(y)}e^{ - \frac{\alpha}{\mu_c(y)}}  \sum_{z \in \Clw{y} }  p(y|z) e^{  - \frac{\alpha}{\mu(y_c|z)} }
\end{split}
\end{equation*}

By first order condition, we know that for an allocation to be optimal, the partial derivative with respect to the rate allocation of agents with positive positive value should converge to the same constant. In other word, by first order condition, we have for all agent $y,y' \in C$ such that $\mu_c(y),\mu_c(y') > 0$, we have
\begin{equation}\label{eq:core_foc}
\frac{dU(y_c|\mu(y_c),\muc_p)}{d\mu_c(y)} = \frac{dU(y_c|\mu(y_c),\muc_p)}{d\mu_c(y')}, \forall \mu_c(y'), \mu_c(y)>0
\end{equation}

The next lemma shows that beneficial rank of content produced by periphery agents. It shows that agent close to the center of community produce content of higher interest to the community.

\begin{lemma}\label{lemma:production_rank}
	If $y,y' \in C$ are two periphery agents such that
	$$||y - \Cc || < ||y' - \Cc ||$$
	then we have that
	$$\sum_{z \in \Clw{y}} p(y|z)> \sum_{z \in \Clw{y'}} p(y'|z).$$
\end{lemma}

\begin{proof}
	Let $y, y'$ be two periphery agents given in the Lemma. Let's consider the difference between $\sum_{z \in \Clw{y}} p(y|z)$ and $\sum_{z \in \Clw{y'}} p(y'|z)$
	by definition we have that 
	\begin{equation*}
	\begin{aligned}
	&\sum_{z \in \Clw{y}} p(y|z) - \sum_{z \in \Clw{y'}} p(y'|z)\\
	=&\Bsbl \sum_{z \in C_p \backslash \{y\}} f(||z-y||) -   \sum_{z \in C_p \backslash \{y'\}} f(||z-y'||) \Bsbr
	\end{aligned}
	\end{equation*}

	Then the result of this lemma follow immediately from the facts that agent $y$ is closer to the center of community than agent $y'
	$ and $f(.)$ is a decreasing concave function, we have that

\end{proof}

In the next lemma we establish that if the periphery agents allocation has the property in the proposition, then the unique optimal respond of core agent has property in the proposition.

\begin{lemma}\label{lemma:core_respond}
	If the allocation of periphery agent $\muc_p$ has the property that
	for agents $y_1, y_2$ in the community such that
	$$\mu(y_c|y_1) > \mu(y_c|y_2)$$
	if 
	$$\|y_1 - \Icen \| < \|y_2 - \Icen \|$$
	then, we have
	$$\mu_c(y_1) > \mu_c(y_2)$$
\end{lemma}

\begin{proof}
	Let $y_1$, $y_2$ be two periphery agents such that 
	$$\|y_1-\Icen \| < \|y_2-\Icen \|$$
	To show the result of this lemma, it is enough to show the following
	\begin{equation*}
	\sum_{z \in \Clw{y_1}} p(y_1|z)e^{\frac{-\alpha}{\mu(y_c|z)}}> \sum_{z \in \Clw{y_2}} p(y_2|z)e^{\frac{-\alpha}{\mu(y_c|z)}}
	\end{equation*}

	then, first order condition condition given in Eq.~\eqref{eq:core_foc} would imply the optimal allocation has 
	$$\mu_c(y_1) > \mu_c
	(y_2)$$
	
	The result above follows immediately from from  Lemma~\ref{lemma:production_rank}.

\end{proof}

\subsection{Periphery Agent Allocation}
Here, we established that the periphery agent optimization problem is concave under Assumption \ref{ass:concave_assumption}, and structural property of periphery agent's optimal respond when core agent's allocation has property described in the proposition. We will use these result to show that if the allocation of core agents satisfy certain property, then it is uniquely optimal for the core agent to respond in a way with structure described in the Proposition \ref{prop:stable_allocation}.
Let's first recall that the utility of periphery agent is given as follows
\begin{equation*}
\begin{aligned}
U_p(y|\mu_c,\mu_p(y))& =
 \sum_{z \in \Clw{y}} \delayCoreUtility{z}{y}  \\
&+ \sum_{z \in \Clw{y}} \delayUtility{z}{y} \\
&+ \alternative{\lambda(y)}.
\end{aligned}
\end{equation*}
The optimization problem of periphery agent is given as
$$\max_{\mu(y)} U_p(y|\mu_c,\mu_p(y))$$
subject to
$$ \sum_{z \in \Clw{y}} \mu(z|y)+ \mu(y_c|y) + \lambda(y) \leq \leafBudget$$
$$\mu(z|y), \lambda(y), \mu(y_c|y) \geq 0.$$

To simplify the notation, let's define
\begin{equation*}
\begin{aligned}
	& \coreUtility{y} =\\
	& \sum_{z \in \Clw{y}} \delayCoreUtility{z}{y}
\end{aligned}
\end{equation*}

The next lemma show that the objective function of periphery agent is concave
\begin{lemma}\label{lemma:leaf_concavity}
	Under Assumption \ref{ass:concave_assumption}, the optimization problems that characterize core periphery agent's allocation is concave.
\end{lemma}
\begin{proof}
	
	Recall that $U_p(y|\mu_c,\mu(y))$ is the objective function in the optimization problem above, and it can be rewritten as 
	\begin{equation*}
	\begin{aligned}
	&U_p(y|\mu_c,\mu(y)) = \\
	& \sum_{z \in \Clw{y}}  \Bsbl f_1(\mu_c(z), \mu(y_c|y)) + f_2(\mu(z|y)) \Bsbr   + f_3(\lambda(y)),
	\end{aligned}
	\end{equation*}
	
	where $f_1(\mu_c(z), \mu(y_c|y))$ is given in Eq.~\eqref{eq:core_term}, $f_2(\mu(z|y))$ is given in Eq.~\eqref{eq:leaf_term} and $f_3(\lambda(y))$ is given in Eq.~\eqref{eq:alternative_term}.  We proved in Proposition~\ref{prop:stable_allocation} that $f_1(\mu_c(z), \mu(y_c|y))$, $f_2(\mu(z|y))$, and $f_3(\lambda(y))$ are indeed concave under Assumption~\ref{ass:concave_assumption}. Therefore $U(y_c|\muc_p,\coreRateVec)$ is sum of concave function, and therefore concave. This completes the proof of this lemma.
\end{proof}

Similarly, the optimal allocation of periphery agent has to satisfy the first order optimal condition, and the partial derivative of positive allocation are computed as following
$$\frac{d\leafOptUtility{y}}{d\mu(y_c|y)}=\frac{\alpha r_p}{\mu^2(y_c|y)}e^{ - \frac{\alpha}{\mu(y_c|y)}}  \sum_{z \in \Clw{y}}  p(z|y) e^{  - \frac{\alpha}{\mu_c(z)} },$$
and
$$\frac{d\leafOptUtility{y}}{d\mu(z|y)} = \frac{\alpha r_p}{\mu^2(z|y)}e^{ - \frac{\alpha}{\mu(z|y)}}   p(z|y),$$
as well as
$$\frac{d\leafOptUtility{y}}{d\lambda(y)}=\firstOrderAlternative{\lambda(y)}.$$

By first order condition, we have 
\begin{equation}\label{eq:leaf_foc}
\begin{aligned}
\frac{d\leafOptUtility{y}}{d\mu(y_c|y)} 
= \frac{d\leafOptUtility{y}}{d\lambda(y)}
= \frac{d\leafOptUtility{y}}{d\mu(z|y)}, \quad \forall \mu(y_c|y), \mu(z|y), \lambda(y)> 0.
\end{aligned}
\end{equation}

In the next lemma, we establish that if periphery agent close to the center of community, then he is going to be more interested in the content of the community.

\begin{lemma}\label{lemma:benefit_rank}
	If $y,y' \in C_p$ are two periphery agents such that
	$$||y - \Cc || < ||y' - \Cc ||$$
	then we have that
	$$\sum_{z \in C_p \backslash \{y\}} p(z|y)> \sum_{z \in C_p \backslash \{y'\}} p(z|y').$$
\end{lemma}

\begin{proof}
	Similar to Lemma~\ref{lemma:production_rank}, the reuslt of this lemma follows immediately from the facts that $y$is closer to the center of community than $y'$and that $p(.|.)$ is concave and decreasing.

\end{proof}

In next lemma, we establish that if the core agents allocation has the property in the proposition, then the unique optimal respond for the periphery agent has property in the proposition.
\begin{lemma}\label{lemma:periphery_respond}
	If the allocation of core agent $y_c$ has the property that
	for agents $y_1, y_2$ in the community such that
	$$\mu_c(y_1) > \mu_c(y_2)$$
	if 
	$$\|y_1 - \Icen \| < \|y_2 - \Icen \|$$
	then, we have
	$$\mu(y_c|y_1) > \mu(y_c|y_2)$$
\end{lemma}

\begin{proof}
	Let $y_1$, $y_2$ be two arbitrary agents such that 
	$$\|y_1-\Icen \| < \|y_2-\Icen \|$$
	To show the result of this lemma, it is enough to show the following
	\begin{equation*}
		S(y_c|y_1) > S(y_c|y_2)
	\end{equation*}
	or equivalently
	\begin{equation}\label{eq:sufficient_condition}
	\sum_{z \in \comSet \backslash \{y_1\}} p(z|y_1)e^{\frac{-\alpha}{\mu(z|y_c)}}> \sum_{z \in \comSet \backslash \{y_2\}} p(z|y_2)e^{\frac{-\alpha}{\mu(z|y_c)}}
	\end{equation}
	then, first order condition condition as given in Eq.~\eqref{eq:leaf_foc}would imply the optimal allocation has 
	$$\mu(y_c|y_1) > \mu(y_c|y_1)$$
	
	Therefore, it remains for to show that 
	\begin{equation*}
	S(y_c|y_1) > S(y_c|y_2)
	\end{equation*}
	
	The result follows immediately from Lemma \ref{lemma:benefit_rank}.

\end{proof}

\subsection{Proof of Proposition \ref{prop:structure}}
Here, we show that an iterative update process, that converge to the unique stable allocation, has the structural properties given in the proposition would preserve in each iteration. Therefore, the stable allocation in the end has the same structural property. 
\begin{proof}
	Let's consider the initial allocation of core agent in the following configuration:
	\begin{equation}
	\mu_{c,0}(y_c) = \{\mu_c(y) = \mu_0\}_{y \in C_p}
	\end{equation}

	Given two periphery agents $y_1,y_2$ such that 
	$$\|y_1-\Icen \| < \|y_2-\Icen \|$$
	the difference between the utility $y_1, y_2$ would benefit from the content of core agent is given by 
	
	\begin{equation}
	\begin{aligned}
	&\sum_{z \in C_p \backslash \{y_1\}} p(z|y_1)e^{\frac{-\alpha}{u_0}} -  \sum_{z \in C_p \backslash \{y_2\}} p(z|y_2)e^{\frac{-\alpha}{u_0}} \\
	& = e^{\frac{-\alpha}{u_0}} \Bsbl \sum_{z \in C_p \backslash \{y_1\}}  p(z|y_1) - \sum_{z \in C_p \backslash \{y_2\}} p(z|y_2) \Bsbr \\
	\end{aligned}
	\end{equation}
	By Lemma~\ref{lemma:production_rank}, we have that 
	$$\sum_{z \in C_p \backslash \{y_1\}}  p(z|y_1) > \sum_{z \in C_p \backslash \{y_2\}} p(z|y_2).$$
	It follows that 
	$$\sum_{z \in C_p \backslash \{y_1\}} p(z|y_1)e^{\frac{-\alpha}{u_0}} -  \sum_{z \in C_p \backslash \{y_2\}} p(z|y_2)e^{\frac{-\alpha}{u_0}} > 0.$$
	Therefore, under the allocation of $\mu_{c,0}(y_c)$, to satisfy the first order condition given in Eq. \eqref{eq:leaf_foc}, the respond $\muc_{l,0}$ of periphery agents would have the following property.
	\begin{equation*}
		\mu(y_c|y_1) > \mu(y_c|y_2)
	\end{equation*}
	
	In Lemma \ref{lemma:core_respond}, we show that if the periphery agents allocation have structure above, it is uniquely optimal for the core agent to have have the structure property as follows.
	\begin{equation*}
		\mu_c(y_1) > \mu_c(y_2)
	\end{equation*}

	Therefore, in the next iterated update of core agent's allocation, $\mu^1(y_c)$, we have 
	\begin{equation*}
		\mu_c(y_1) > \mu_c(y_2)
	\end{equation*}

	In Lemma \ref{lemma:periphery_respond}, we show that if the core agent allocation has structure above, then the structure of periphery agents' allocation would preserved,i.e, 
	\begin{equation*}
	\mu(y_c|y_1) > \mu(y_c|y_2)
	\end{equation*}

	Now it remains to show that at each iteration, the total utility is monotonic increasing. Notice that the iteration process is initiated by core agent and by formulation of his objective function, the core agent would only change allocation if the new allocation would increase the total utility of community. This means that the total utility would increase after core agent's update,i.e,
	\begin{equation*}
	G((\mu^{k+1}(y_c),\muc^k_p)) \geq G((\mu^{k}(y_c),\muc^k_p))
	\end{equation*}
	
	Again, by formulation of the periphery agent's objective function, the periphery agent would only change allocation if the new allocation would increase his own utility and such change would not effect other periphery agents' allocation.
	
	\begin{equation*}
	G((\mu^{k+1}(y_c),\muc^{k+1}_p)) \geq G((\mu^{k+1}(y_c),\muc^k_p))
	\end{equation*}
	
	These means, after each update, the utility is monotonic increasing (either increase or stay the same).
	\begin{equation*}
		G(\muc^{k+1}) \geq G(\muc^{k})  
	\end{equation*}
	
	As we show in Lemma \ref{lemma:potential_game_convexity}, the $G(\muc)$ is concave. This means it has a unique optimal allocation $\muc^*$.
	
	Therefore, combing the result above, we get that 
	\begin{equation*}
		\lim_{k \rightarrow \infty}G(\muc^{k}) = G(\muc^*)  
	\end{equation*}
	
	Therefore, we get that the stable allocation(Nash Equilibrium) of core-periphery interaction has structure in the proposition.

\end{proof}

\end{document}